\documentclass[11pt,dvipsnames,letterpaper]{article}

\usepackage{privacy}

\begin{document}

\title{Private Frequency Estimation via Projective Geometry}
\author{Vitaly Feldman\thanks{Apple Inc.} \and Jelani Nelson\thanks{UC Berkeley. \texttt{minilek@berkeley.edu}. Supported by NSF grant CCF-1951384, ONR grant N00014-18-1-2562, and ONR DORECG award N00014-17-1-2127.} \and Huy L. Nguyen\thanks{Northeastern University. \texttt{hu.nguyen@northeastern.edu}. Supported in part by NSF CAREER grant CCF-1750716 and NSF grant CCF-1909314.} \and Kunal Talwar\thanks{Apple Inc. \texttt{ktalwar@apple.com}.}}
\date{\today}
\maketitle

\begin{abstract}
  In  this work, we propose a new algorithm \projectivegeometry (\pg) for locally differentially private (LDP) frequency estimation. For a universe size of $k$ and with $n$ users, our $\eps$-LDP algorithm has communication cost $\lceil\log_2k\rceil$ bits in the private coin setting and $\eps\log_2 e + O(1)$ in the public coin setting, and has computation cost $O(n + k\exp(\eps) \log k)$ for the server to approximately reconstruct the frequency histogram, while achieving the state-of-the-art privacy-utility tradeoff. In many parameter settings used in practice this is a significant improvement over the $ O(n+k^2)$ computation cost that is achieved by the recent \pirappor algorithm (Feldman and Talwar; 2021). Our empirical evaluation shows a speedup of over 50x over \pirappor while using approximately 75x less memory for practically relevant parameter settings. In addition, the running time of our algorithm is within an order of magnitude of \hadamardresponse (Acharya, Sun, and Zhang; 2019) and \recursivehadamardresponse (Chen, Kairouz, and Ozgur; 2020) which have significantly worse reconstruction error. The error of our algorithm essentially matches that of the communication- and time-inefficient but utility-optimal \subsetselection (\SS) algorithm (Ye and Barg; 2017).   
  Our new algorithm is based on using Projective Planes over a finite field to define a small collection of sets that are close to being pairwise independent and a dynamic programming algorithm for approximate histogram reconstruction on the server side.
  We also give an extension of \pg, which we call \hybridprojectivegeometry, that allows trading off computation time with utility smoothly.
\end{abstract}

\section{Introduction}

In the so-called {\it federated} setting, user data is distributed over many devices which each communicate to some central server, after some local processing, for downstream analytics and/or machine learning tasks. We desire such schemes which (1) minimize communication cost, (2) maintain privacy of the user data while still providing utility to the server, and (3) support efficient algorithms for the server to extract knowledge from messages sent by the devices. Such settings have found applications to training language models for such applications as autocomplete and spellcheck, and other analytics applications in Apple iOS \cite{ThakurtaVV+17} and analytics on settings in Google Chrome \cite{ErlingssonPK14}.

The gold standard for protecting privacy is for a scheme to satisfy {\it differential privacy}. In the so-called {\it local model} that is relevant to the federated setting, there are $n$ users with each user $i$ holding some data $d_i\in\mathcal D$. Each user then uses its own private randomness $r_i$ and data $d_i$ to run a {\it local randomizer} algorithm that produces a random message $M_i$ to send to the server.  We say the scheme is {\it $\eps$-differentially private} if for all users $i$, any possible message $m$, and any $d\neq d'$,
$$
\Pr(M_i = m | d_i = d) \le e^\eps \Pr(M_i = m | d_i = d') .
$$
Note a user could simply send an unambiguous encoding of $d_i$, which allows the server to learn $d_i$ exactly (perfect utility), but privacy is not preserved; such a scheme does not preserve $\eps$-DP for any finite $\eps$. On the opposite extreme, the user could simply send a uniformly random message that is independent of $d_i$, which provides zero utility but perfect privacy ($\eps = 0$). One can hope to develop schemes that smoothly increase utility by relaxing privacy (i.e., by increasing $\eps$).

This work addresses the problem of designing efficient schemes for locally differentially private frequency estimation. In this problem, one defines a histogram $x\in\R^{k}$ where $x_d$ is the number of users $i$ with $d_i = d$, and $k = |\mathcal D|$. From the $n$ randomized messages it receives, the server would like to approximately reconstruct the histogram, i.e., compute some $\tilde x$ such that $\|x - \tilde x\|$ is small with good probability over the randomness $r = (r_1,\ldots,r_n)$, for some norm $\|\cdot\|$. Our goal is to design schemes that obtain the best-known privacy-utility trade-offs, while being efficient in terms of communication, computation time, and memory. In this work we measure {\it utility loss} as the mean squared error (MSE) $\E_r\frac 1k[\|x - \tilde x\|_2^2]$, with lower MSE yielding higher utility. Note that such a scheme should specify both the local randomizer employed by users, and the reconstruction algorithm used by the server.

\begin{table*}[!h]
  \begin{center}
  \begin{tabular}{|c|c|c|c|}
    \hline
    \textbf{scheme name} & \textbf{communication} & \textbf{utility loss} & \textbf{server time}\\
    \hline
    \randomizedresponse&$\lceil \log_2 k\rceil$&$\frac{n(2e^\eps + k)}{(e^\eps - 1)^2}$&$n+k$\\
    \hline
    \rappor \cite{ErlingssonPK14}& $k$ & $\frac{4ne^{\eps}}{(e^{\eps} - 1)^2}$& $nk$\\
    \hline
    \subsetselection \cite{YeB17,WangHNZWXY19}&$\frac k{e^\eps}(\eps + O(1))$&$\frac{4ne^{\eps}}{(e^{\eps} - 1)^2}$&$ n\frac{k}{e^\eps}$\\
    \hline
    \pirappor \cite{FeldmanT21}& $\lceil\log_2 k\rceil + O(\eps)$&$\frac{4ne^{\eps}}{(e^{\eps} - 1)^2}$& $\min(n + k^2, n \frac{k}{e^\eps})$, or\\
    &&&$n + ke^{2\eps}\log k$ ({\it this work})\\
    \hline
    \hadamardresponse \cite{AcharyaSZ19}&$\lceil \log_2 k\rceil$&$\frac{36ne^{\eps}}{(e^{\eps} - 1)^2}$&$n + k\log k$\\
    \hline
    \recursivehadamardresponse \cite{ChenKO20}&$\lceil \log_2 k\rceil$
    &$\frac{8ne^{\eps}}{(e^{\eps} - 1)^2}$
    &$n+k\log k$\\
    \hline
    \hline
    \projectivegeometry &$\lceil\log_2 k\rceil$&$\frac{4ne^{\eps}}{(e^{\eps} - 1)^2}$&$n + ke^\eps\log k$\\
    \hline
    \hybridprojectivegeometry  &$\lceil\log_2 k\rceil$&$(1+\frac{1}{q-1})\frac{4ne^{\eps}}{(e^{\eps} - 1)^2}$&$n + kq\log k$\\
    \hline
  \end{tabular}
  \caption{Known local-DP schemes for private frequency estimation compared with ours. Utility bounds are given up to $1+o_k(1)$ multiplicative accuracy for ease of display and running times are asymptotic. For brevity we only state bounds for $\eps \leq \log k$. Some of algorithms assume $k$ is either a power of $2$ or some other prime power and otherwise potentially worsen in some parameters due to round-up issues; we ignore this issue in the table. The communication and server time for \rappor are random variables which are never more than $k$ and $nk$, respectively, but \rappor can be implemented so that in expectation the communication and runtimes are asymptotically equal to \subsetselection. For \hybridprojectivegeometry, $q$ can be chosen as any prime in $[2,\exp(\eps)+1]$. The utility loss here is the proven upper bound on the variance for \pg, \hr and \rhr, and the analytic expression for the variance for the others. The communication bounds are in the setting of private coin protocols. As with \rhr, \pg and \hpg can also both achieve improved communication in the public coin model; see \cref{sec:pub-coin}.}\label{fig:algs}
  \end{center}
\end{table*}

There are several known algorithms for this problem; see \cref{fig:algs}. To summarize, the best known utility in prior work is achieved by \textsf{SubsetSelection} and slightly worse utility is achieved by the \textsf{RAPPOR} algorithm \cite{ErlingssonPK14} that is based on the classical binary randomized response \cite{Warner65}.  Unfortunately, both RAPPOR and Subset Selection have very high communication cost of $\approx k H(1/(e^\eps+1))$, where $H$ is the binary entropy function and server-side running time of $\tilde O(n k/\exp(\eps))$.  Large $k$ is common in practice, e.g., $k$ may be the size of a lexicon when estimating word frequencies to train language models. This has led to numerous and still ongoing efforts to design low-communication protocols for the problem \cite{hsu2012distributed,ErlingssonPK14,bassily2015local,kairouz2016discrete,WangHNZWXY19,WangBLJ:17,YeB17,AcharyaSZ19,bun2019heavy,bassily2020practical,ChenKO20,FeldmanT21,ShahCBKT21}.
 
 One simple approach to achieve low communication and computational complexity is to use a simple $k$-ary \textsf{RandomizedResponse} algorithm (e.g.~\cite{WangBLJ:17}). Unfortunately, its utility loss is suboptimal by up to an $\Omega(k/e^\eps)$ factor; recall $k$ is often large and $\eps$ is at most a small constant, and thus this represents a large increase in utility loss. In the $\eps < 1$ regime asymptotically optimal utility bounds are known to be achievable with low communication and computational costs \cite{bassily2015local,bun2019heavy,bassily2020practical}. 
 The first low-communication algorithm that achieves asymptotically optimal bounds in the $\eps > 1$ regime is given in \cite{WangBLJ:17}. It communicates $O(\eps)$ bits and relies on shared randomness. However, it matches the bounds achieved by \textsf{RAPPOR} only when $e^\eps$ is an integer and its computational cost is still very high and comparable to that of RAPPOR. Two algorithms, \textsf{HadamardResponse}~\cite{AcharyaSZ19} and \textsf{RecursiveHadamardResponse}~\cite{ChenKO20}, show that it is possible to achieve low communication, efficient computation (only $\Theta(\log k)$ slower than \textsf{RandomizedResponse}) and {\it asymptotically} optimal utility. However, their utility loss in practice is suboptimal by a constant factor (e.g.~our experiments show that these algorithms have an MSE that is  over $2\times$ higher for $\eps =5$ than \textsf{SubsetSelection}; see \cref{fig:experiments}).
 
Recent work of Feldman and Talwar \cite{FeldmanT21} describes a general technique for reducing communication of a local randomizer without sacrificing utility and, in particular, derives a new low communication algorithm for private frequency estimation via pairwise independent derandomization of \textsf{RAPPOR}. Their algorithm, referred to as \textsf{PI-RAPPOR}, achieves the same utility loss as \textsf{RAPPOR} and has the server-side running time of  $\tilde O(\min(n + k^2, nk/\exp(\eps)))$. The running time of this algorithm is still prohibitively high when both $n$ and $k$ are large. 

We remark that while goals (1)-(3) from the beginning of this section are all important, goal (2) of achieving a good privacy/utility tradeoff is unique in that poor performance cannot be mitigated by devoting more computational resources (more parallelism, better hardware, increased bandwidth, etc.). After deciding upon a required level of privacy $\eps$, there is a fundamental limit as to how much utility can be extracted given that level of privacy; our goal in this work is to understand whether that limit can be attained in a communication- and computation-efficient way.

\paragraph{Our main contributions.} We give a new private frequency estimation algorithm \projectivegeometry (\pg) that maintains the best known utility and low communication while significantly improving computational efficiency amongst algorithms with similarly good utility. Using our ideas, we additionally give a new reconstruction algorithm that can be used with the \pirappor mechanism to speed up its runtime from $O(k^2/\exp(\eps))$ to $O(k\exp(2\eps)\log k)$ (albeit, this runtime is still slower than \pg's reconstruction algorithm by an $\exp(\eps)$ factor). We also show a general approach that can further improve the server-side runtime at the cost of slightly higher reconstruction error, giving a smooth tradeoff: for any prime $2\le q \le \exp(\eps)+1$, we can get running time $O(n+ q k\log k)$ with error only $(1+1/(q-1))$ times larger than the best known bound\footnote{For both \pg and \hpg we have stated runtime bounds assuming that certain quantities involving $k, \exp(\eps)$ are prime powers. If this is not the case, runtimes may increase by a factor of $\exp(\eps)$ for \pg, or $q$ for \hpg; we note that \pirappor also has this feature.}. 
Note that for $q=2$ we recover the bounds achieved by \hr and \rhr. Our mechanisms require $\lceil \log_2 k\rceil$ per device in the private coin model, or $\eps\log_2 e + O(1)$ bits in the public coin model (see \cref{sec:pub-coin}). As in previous work, our approximate reconstruction algorithm for the server is also parallelizable, supporting linear speedup for any number of processors $P\le \min\{n, k\exp(\eps)\}$. 
We also perform an empirical evaluation of our algorithms and prior work and show that indeed the error of our algorithm matches the state of the art will still being time-efficient. 

\medskip

As has been observed in previous work~\cite{AcharyaSZ19}, the problem of designing a local randomizer is closely related to the question of existence of set systems consisting of sets of density $\approx \exp(-\eps)$ which are highly symmetric, and do not have positive pairwise dependencies. The size of the set system then determines the communication cost, and its structural properties may allow for efficient decoding. We show that projective planes over finite fields give us set systems with the desired properties, leading to low communication and state-of-the-art utility. We also show a novel dynamic programming algorithm that allows us to achieve server runtime that is not much worse than the fastest known algorithms.

As in a lot of recent work on this problem, we have concentrated on the setting of moderately large values for the local privacy parameter $\eps$. This is a setting of interest due to recent work in privacy amplification by shuffling~\cite{BittauEMMR17, CheuSUZZ19, ErlingssonFMRTT19, BalleBGN19, FeldmanMcTa20} that shows that local DP responses, when shuffled across a number of users so that the server does not know which user sent which messages, satisfy a much stronger central privacy guarantee. Asymptotically, $\eps$-DP local randomizers aggregated over $n$ users satisfy $(O(\sqrt{\frac{e^{\eps}\ln \frac 1 \delta}{n}}), \delta)$-DP. The hidden constants here are small: as an example with $n=10,000$ and $\eps=6$, shuffling gives a central DP guarantee of $(0.3, 10^{-6})$-DP. This motivation from shuffling is also the reason why our work concentrates on the setting of private coin protocols, as shared randomness seems to be incompatible with shuffling of private reports.
We note that while constant factors improvement in error may seem small, these algorithm are typically used for discovering frequent items from power law distributions. A constant factor reduction in variance of estimating any particular item frequency then translates to a corresponding smaller noise floor (for a fixed false positive rate, say), which then translates to a constant factor more items being discovered.

\subsection{Related Work}
\label{sec:related}
A closely related problem is finding ``heavy hitters'', namely all elements $j \in [k]$ with counts higher than some given threshold; equivalently, one wants to recover an approximate histogram $\tilde x$ such that $\|x - \tilde x\|_\infty$ is small (the non-heavy hitters $i$ can simply be approximated by $\tilde x_i = 0$). In this problem the goal is to avoid linear runtime dependence on $k$ that would result from doing frequency estimation and then checking all the estimates. This problem is typically solved using a ``frequency oracle'' which is an algorithm that for a given $j\in [k]$ returns an estimate of the number of $j$'s held by users (typically without computing the entire histogram) \cite{bassily2015local,bassily2020practical,bun2019heavy}. Frequency estimation is also closely related to the discrete distribution estimation problem in which inputs are sampled from some distribution over $[k]$ and the goal is to estimate the distribution \cite{YeB17,AcharyaSZ19}. Indeed, bounds for frequency estimation can be translated directly to bounds on distribution estimation by adding the sampling error. We note that even for the problem of implementing a private frequency oracle, our \pg scheme supports answering queries faster than \pirappor by factor of $\Theta(\exp(\eps))$.

\global\long\def\Ex{\mathbb{E}}%
\global\long\def\one{\mathbf{1}}%
\global\long\def\K{\mathcal{K}}%
\global\long\def\F{\mathcal{\mathbb{F}}}%
\global\long\def\R{\mathbb{R}}%
\global\long\def\Z{\mathbb{Z}}%


\section{Preliminaries}

Our mechanisms are based on projective spaces, and below we review
some basic definitions and constructions of such spaces from standard vector spaces.
\begin{definition}
For a given vector space $V$, the projective space $P\left(V\right)$
is the set of equivalence classes of $V\setminus\left\{ 0\right\} $,
where $0$ denotes the zero vector, under the following equivalence
relation: $x\sim y$ iff $x=cy$ for some scalar $c$. Each equivalence
class is called a \emph{(projective) ``point''} of the projective
space. Let $p:V\setminus\left\{ 0\right\} \to P\left(V\right)$ be
the mapping from each vector $v\in V$ to its equivalence class. If
$V$ has dimension $t$ then $P(V)$ has dimension $t-1$.
\end{definition}

We will also use subspaces of the projective space $P\left(V\right)$.
\begin{definition}
A projective subspace $W$ of $P\left(V\right)$ is a subset of $P\left(V\right)$
such that there is a subspace $U$ of $V$ where $p\left(U\setminus\left\{ 0\right\} \right)=W$.
If $U$ has dimension $t$ then $W$ has dimension $t-1$.
\end{definition}

It should be noted that intersections of projective subspaces are
projective subspaces. Let $q$ be a prime power and $\F_{q}^{t}$
the $t$-dimensional vector space over the field $\F_{q}$. We will
work with $P\left(\F_{q}^{t}\right)$ and its subspaces.
\begin{definition}
A vector $x\in\F_{q}^{t}$ is called canonical if its first non-zero
coordinate is $1$.
\end{definition}

Each equivalence class can be specified by its unique canonical member.

\section{\projectivegeometry description and analysis}\label{sec:basic-construction}

Our \pg scheme is an instantiation of the framework due to~\cite{AcharyaSZ19}.
In their framework, the local randomizer is implemented as follows. There is a 
universe $U$ of outputs and each input $v$ corresponds to a subset $S(v)$ of outputs.
All the subsets $S(v)$ for different values of $v$ have the same size.
Given the input $v$, the local randomizer returns a uniformly random element of
$S(v)$ with probability $\frac{e^{\varepsilon}|S(v)|}{|S(v)|e^{\varepsilon}+|U|-|S(v)|}$
and a uniformly random element of $U\setminus S(v)$ with probability 
$\frac{|U|-|S(v)|}{|S(v)|e^{\varepsilon}+|U|-|S(v)|}$. The crux of the construction is
in specifying the universe $U$ and the subsets $S(v)$.

\pg works for $k=\frac{q^t-1}{q-1}$ for some integer $t$ (other
values of $k$ need to be rounded up to the nearest such value). We identify the $k$
input values with $k$ canonical vectors in $\F_{q}^{t}$ and the corresponding 
projective points in $P\left(\F_{q}^{t}\right)$.
We also identify the output values with projective points in $P\left(\F_{q}^{t}\right)$.
The subsets $S(v)$ are the $(t-2)$-dimensional projective subspaces of $P\left(\F_{q}^{t}\right)$.
There are $\frac{q^{t}-1}{q-1}$ $(t-2)$-dimensional projective subspaces, which is the
same as the number of projective points. For a canonical
vector $v$, the set $S(v)$ is the $(t-2)$-dimensional projective subspace
 such that for all $u\in p^{-1}(S(v))$, we have $\left\langle u,v\right\rangle =0$.
Each $(t-2)$-dimensional projective subspace contains
$\frac{q^{t-1}-1}{q-1}$ projective points. In other words, each set $S(v)$ contains 
$\frac{q^{t-1}-1}{q-1}$ messages out of the universe of $\frac{q^{t}-1}{q-1}$ messages.

An important property of the construction is the symmetry among the intersections of any two 
subsets $S(v)$.
\begin{claim}
Consider a $t$-dimensional vector space $V$.
The intersection of any two 
$(t-2)$-dimensional projective subspaces of $P\left(V\right)$ is a $(t-3)$-dimensional projective 
subspace.
\end{claim}
\begin{proof}
Let $I$ be the intersection of two projective subspaces $S_1$ and $S_2$. Recall that $I$, $S_1$,
$S_2$ are projective subspaces corresponding to subspaces of $V$.
Assume for contradiction that the
dimension $d-1$ of the intersection $I$ is lower than $t-3$. Starting
from a basis $v_{1},\ldots,v_{d}$ of $p^{-1}(I)\cup\left\{ 0\right\}$,
we can extend it with $u_{1},\ldots,u_{t-1-d}$ to form a basis of
the subspace $p^{-1}(S_1)\cup\left\{ 0\right\}$. We can
also extend $v_{1},\ldots,v_{d}$ with $w_{1},\ldots,w_{t-1-d}$ to
form a basis of $p^{-1}(S_2)\cup\left\{ 0\right\}$. 
Because $d+2(t-1-d)=t+(t-2-d)>t$, the collection of vectors
$v_{1},\ldots,v_{d},u_{1},\ldots,u_{t-1-d},w_{1},\ldots,w_{t-1-d}$
must be linearly dependent. There must exist nonzero coefficients
so that $\sum_{i}\alpha_{i}v_{i}+\sum_{j}\beta_{j}u_{j}+\sum_{k}\gamma_{k}w_{k}=0$.
This means $\sum_{k}\gamma_{k}w_{k}=-\sum_{i}\alpha_{i}v_{i}-\sum_{j}\beta_{j}u_{j}$
is a non-zero vector in $p^{-1}(S_1)\cap p^{-1}(S_2)$
but it is not in $p^{-1}(I)$, which is a contradiction.
\end{proof}

To ease the presentation we define $c_{int}=\frac{q^{t-2}-1}{q-1}$ 
to be the size of the intersection of two subsets $S(v)$ and let
$c_{set}=\frac{q^{t-1}-1}{q-1}$ denote the size of each subset $S(v)$.
Notice that $c_{set}^{2}\ge k\cdot c_{int}$
i.e. $(c_{set}/c_{int})^{2}\ge k/c_{int}$.

Each user with input $v$ sends
a projective point $e$ with probability $e^{\varepsilon}p$ if $e$
is in $S(v)$ and probability $p$ otherwise. We have

\begin{align*}
e^{\varepsilon}pc_{set}+p(k-c_{set}) & =1,\\
\mbox{ so that }p & =\frac{1}{\left(e^{\varepsilon}-1\right)c_{set}+k}.
\end{align*}

The server keeps the counts on the received projective points in a
vector $y\in \Z^k$. Thus, the total server storage is $O\left(k\right)$. We estimate
$x_{v}$ by computing
\[
\tilde{x}_{v}=\alpha\left(\sum_{u\in S_{v}}y_{u}\right)+\beta\sum_{u}y_{u}
\]
where $\alpha$ and $\beta$ are chosen so that it is an unbiased estimator. Note $\sum_u y_u = n$. We would like $\E \tilde x_v = x_v$ for all $v$. Notice that by linearity of expectation, it suffices to focus on the contribution to $\tilde x_v$ from a single user.

If that user's input is $v$, the expectation of the sum $Q:=\sum_{u\in S_{v}}y_{u}$ is $e^{\varepsilon}pc_{set}$. On the other hand, if the input is not $v$, the expectation of the sum $\sum_{u\in S_{v}}y_{u}$ is $e^{\varepsilon}pc_{int}+p\left(c_{set}-c_{int}\right)$.
We want $\alpha\cdot\Ex\left[Q\right]+\beta=[[input\ is\ v]]$, where $[[T]]$ is defined to be $1$ if $T$ is true and $0$ if false. Thus,
\begin{align*}
\alpha e^{\varepsilon}pc_{set}+\beta & =1,\\
\mbox{ and } \alpha p\left(\left(e^{\varepsilon}-1\right)c_{int}+c_{set}\right)+\beta & =0.
\end{align*}

Substituting $p$, we get
\begin{align*}
\alpha e^{\varepsilon}\frac{c_{set}}{\left(e^{\varepsilon}-1\right)c_{set}+k}+\beta & =1\\
\alpha\frac{(e^{\varepsilon}-1)c_{int}+c_{set}}{\left(e^{\varepsilon}-1\right)c_{set}+k}+\beta & =0
\end{align*}

Solving for $\alpha, \beta$, we get
\begin{align*}
\alpha & =\frac{\left(e^{\varepsilon}-1\right)c_{set}+k}{\left(e^{\varepsilon}-1\right)(c_{set}-c_{int})};\\
\beta & =-\frac{\left(e^{\varepsilon}-1\right)c_{set}+k}{\left(e^{\varepsilon}-1\right)(c_{set}-c_{int})}\cdot\frac{(e^{\varepsilon}-1)c_{int}+c_{set}}{\left(e^{\varepsilon}-1\right)c_{set}+k}=-\frac{(e^{\varepsilon}-1)c_{int}+c_{set}}{\left(e^{\varepsilon}-1\right)(c_{set}-c_{int})}.
\end{align*}

We next analyze the variance, which suggests that $q$ should be chosen close to $\exp(\eps) + 1$ for the best utility.

\begin{lemma}\label{lem:pg-var}
$\Ex\left[\left\Vert x - \tilde x\right\Vert_2 ^{2}\right]\le\frac{ne^{\varepsilon}c_{set}^2/c_{int}^2+n(k-1)\left((e^{\varepsilon}-1)+c_{set}/c_{int}\right)^{2}}{\left(e^{\varepsilon}-1\right)^{2}(c_{set}/c_{int}-1)}$.
In particular, if $c_{set}/c_{int}=e^{\varepsilon}+1$ then
$\Ex\left[\frac 1k\left\Vert x-\tilde{x}\right\Vert_2 ^{2}\right]\le \frac nk+\frac{4ne^{\varepsilon}}{\left(e^{\varepsilon}-1\right)^{2}}$
\end{lemma}

\begin{proof}
By independence, we only need to analyze the variance when there is
exactly one user with input $v$. The lemma then follows from adding
up the variances from all users.

\begin{align*}
\Ex\left[\left(\tilde{x}_{v}-1\right)^{2}\right] & =e^{\varepsilon}pc_{set}\left(\alpha+\beta-1\right)^{2}+p(k-c_{set})\left(\beta-1\right)^{2}\\
 & =\frac{1-\beta}{\alpha}\left(\alpha+\beta-1\right)^{2}+\frac{\alpha+\beta-1}{\alpha}\left(1-\beta\right)^{2}\\
 & =\left(\alpha+\beta-1\right)\left(1-\beta\right)\\
 & =\frac{-c_{set}+k}{\left(e^{\varepsilon}-1\right)(c_{set}-c_{int})}\cdot\frac{e^{\varepsilon}c_{set}}{\left(e^{\varepsilon}-1\right)(c_{set}-c_{int})}\\
 & =\frac{\left(-c_{set}/c_{int}+k/c_{int}\right)e^{\varepsilon}c_{set}/c_{int}}{\left(e^{\varepsilon}-1\right)^2(c_{set}/c_{int}-1)^2}\\
 & \le\frac{\left(-c_{set}/c_{int}+c_{set}^2/c_{int}^2\right)e^{\varepsilon}c_{set}/c_{int}}{\left(e^{\varepsilon}-1\right)^2(c_{set}/c_{int}-1)^2}\\
 & =\frac{e^{\varepsilon}c_{set}^2/c_{int}^2}{\left(e^{\varepsilon}-1\right)^{2}(c_{set}/c_{int}-1)}
\end{align*}

Let $z=c_{set}/c_{int}$. Note that $\frac{z^2}{z-1}$ is an increasing function for $z\in [2,+\infty)$ so this part of the variance gets larger as
$q$ gets larger.

Next we analyze the contribution to the variance from coordinates
$u\ne v$. 

\begin{align*}
\Ex\left[\tilde{x}_{u}^{2}\right] & =\left(\left(e^{\varepsilon}-1\right)c_{int}+c_{set}\right)p\left(\alpha+\beta\right)^{2}+\left(1-e^{\varepsilon}pc_{int}-p\left(c_{set}-c_{int}\right)\right)\beta^{2}\\
 & =-\frac{\beta}{\alpha}\left(\alpha+\beta\right)^{2}+\left(1+\frac{\beta}{\alpha}\right)\beta^{2}\\
 & =\frac{-\beta\left(\alpha+\beta\right)^{2}+\left(\alpha+\beta\right)\beta^{2}}{\alpha}\\
 & =-\beta\left(\alpha+\beta\right)\\
 & =\frac{(e^{\varepsilon}-1)c_{int}+c_{set}}{\left(e^{\varepsilon}-1\right)(c_{set}-c_{int})}\cdot\frac{\left(e^{\varepsilon}-2\right)c_{set}+k-(e^{\varepsilon}-1)c_{int}}{\left(e^{\varepsilon}-1\right)(c_{set}-c_{int})}\\
 & =\frac{(e^{\varepsilon}-1)+c_{set}/c_{int}}{\left(e^{\varepsilon}-1\right)(c_{set}/c_{int}-1)}\cdot\frac{\left(e^{\varepsilon}-2\right)c_{set}/c_{int}+k/c_{int}-(e^{\varepsilon}-1)}{\left(e^{\varepsilon}-1\right)(c_{set}/c_{int}-1)}\\
 & \le \frac{\left((e^{\varepsilon}-1)+z\right)\left(\left(e^{\varepsilon}-2\right)z+z^2-(e^{\varepsilon}-1)\right)}{\left(e^{\varepsilon}-1\right)^2(z-1)^2}\\
 & =\frac{\left((e^{\varepsilon}-1)+z\right)^2}{\left(e^{\varepsilon}-1\right)^2(z-1)}\\
\end{align*}

Note that the function $\frac{\left((e^{\varepsilon}-1)+z\right)^{2}}{\left(e^{\varepsilon}-1\right)^{2}(z-1)}$
is decreasing for $z\in\left(0,e^{\varepsilon}+1\right]$
and it is increasing for $z\in\left[e^{\varepsilon}+1,+\infty\right)$
so this part of the variance is minimized when $z=e^{\varepsilon}+1$. For $z=e^{\varepsilon}+1$, we can
substitute and get $\frac{4e^{\varepsilon}}{\left(e^{\varepsilon}-1\right)^{2}}$.
\end{proof}

Next we discuss the algorithms to compute $\tilde{x}_{v}$. The
naive algorithm takes $O(kc_{set})=O(k^{2}/q)$ time and this is the
algorithm of choice for $t\le3$. For $t>3$, we can use dynamic programming
to obtain a faster algorithm. Note in the below that $q$ should be chosen close to $\exp(\eps)+1$.

\begin{theorem}
In the \projectivegeometry scheme, there exists an $O((q^t-1)/(q-1)tq)$
time algorithm for server reconstruction, using $O((q^t-1)/(q-1))$ memory. These bounds are at best $O(ktq)$ time and $O(k)$ memory, and increase by at most a factor of $q$ each if rounding up to the next power of $q$ is needed so that $(q^t-1)/(q-1) \ge k$.
\end{theorem}
\begin{proof}
We use dynamic programming. For $a\in\mathbb{F}_{q}^{j},b\in\mathbb{F}_{q}^{t-j},z\in\mathbb{F}_{q}$,
where $a$ is further restricted to have its first nonzero entry be
a $1$ (it may also be the all-zeroes vector), and $b$ is restricted
to be a canonical vector when $j=0$, define

\[
f(a,b,z)=\sum_{\substack{pref_j(u)=a\\
\langle\mathrm{suff}_{t-j}(u),b\rangle=z\\
\\
}
}y_{u},
\]

\noindent where $\mathrm{pref}_{i}(u)$ denotes the length-$i$ prefix
vector of $u$, and $\mathrm{suff}_{i}(u)$ denotes the length-$i$
suffix vector of $u$. Then, we would like to compute

\[
\tilde{x}_{v}=\alpha\left(\sum_{u\in S_{v}}y_{u}\right)+\beta\sum_{u}y_{u}=\alpha\cdot f(\bot,v,0)+\beta n,
\]

\noindent for all projective points $v$, where $\bot$ denotes the length-$0$
empty vector. We next observe that $f$ satisfies a recurrence relation,
so that we can compute the full array of values $(f(\bot,v,0))_{v\text{ is canonical}}$
efficiently using dynamic programming and then efficiently obtain
$\ensuremath{\tilde{x}\in\mathbb{R}^{k}}$.

We now describe the recurrence relation. For $w\in\mathbb{F}_{q}$
and a vector $v$, let $v\circ w$ denote $v$ with $w$ appended
as one extra entry. If $j$ denotes the length of the vector $a$,
then the base case is $j=t$. In this case, $f(a,\bot,z)=y_{a}$ iff
both $a\neq0$ and $z=0$; else, $f(a,\bot,z)=0$. The recursive step
is then when $0\le j<t$. Essentially, we have to sum over all ways
to extend $a$ by one more coordinate. Let $\mathrm{suff_{-1}(b)}$
denote the vector $b$ but with the first entry removed (so it is
a vector of length one shorter). There are two cases: $a$ is the
all-zeroes vector, versus it is not. In the former case, the recurrence
is
\[
f(0,b,z)=f(\vec{0}\circ0,\mathrm{suff_{-1}(b)},z)+f(\vec{0}\circ1,\mathrm{suff_{-1}(b)},z-b_{1}\mod q).
\]

\noindent Note we are not allowed to append $w\in\{2,3,\ldots,q-1\}$
to $a$ since that would not satisfy the requirement that the first
argument to $f$ either be all-zeroes or be canonical. The other case
for the recurrence relation is when $a\neq0$, in which case the recurrence
relation becomes

\[
f(a,b,z)=\sum_{w=0}^{q-1}f(a\circ w,\mathrm{suff_{-1}}(b),z-d\cdot b_{1}\mod q).
\]

We now analyze the running time and memory requirements to obtain
all $f(a,b,z)$ values via dynamic programming. The runtime is proportional to

\[
kq+\sum_{\substack{a,b,z,j\neq0}
}q.
\]

\noindent This is because for $j>0$, for each $a,b,z$ triple we
do at most $q$ work. When $j=0$, there is only one possible value
for $a$ (namely $\bot$) and $k=\frac{q^{t}-1}{q-1}$ values for
$v$, plus we are only concerned with $z=0$ in this case. For larger
$j$, the number of possibilities for $a$ is $\frac{q^{j}-1}{q-1}+1$
(the additive $1$ is since $a$ can be the all-zeroes vector), whereas
the number of possibilities for $b$ is $q^{t-j}$. Thus the total
runtime is proportional to
\[
kq+\left(\sum_{j=1}^{t}\left(\frac{q^{j}-1}{q-1}+1\right)\cdot q^{t-j}\right)\cdot q^{2}=O(ktq^{2}).
\]

For the memory requirement, note $f(\cdot)$ values for some fixed
$j$ only depend on the values for $j+1$, and thus using bottom-up
dynamic programming we can save a factor of $t$ in the memory, for
a total memory requirement of only $O(kq)$ (for any fixed $j$ there
are only $O(k)$ $a,b$ pairs, and there are $q$ values for $z$).

Finally, we add an optimization which improves both the runtime and memory by a factor of $q$. Specifically, suppose $b$ is
not canonical and is not the all-zeroes vector. Let the value of its
first nonzero entry be $\zeta$. Then $f(a,b,z)$ is equal to $f(a,b/\zeta,z/\zeta)$,
where the division is over $\mathbb{F}_{q}$. Thus, we only need to
compute $f(\cdot)$ for $b$ either canonical or equal to the $0$
vector. This reduces the number of $b$ from $q^{t-j}$ to $(q^{t-j}-1)/(q-1)+1$,
which improves the runtime to $O(ktq)$ and the memory to $O(k)$.  Note finite field division over $\mathbb F_q$ can be implemented in $O(1)$ time after preprocessing. First,  factor $q-1$ and generate all its divisors in $o(q)$ time, from which we can find a generator $g$ of $\mathbb F_q^*$ in $o(q)$ expected time by rejection sampling (it is a generator iff $g^p\not\equiv 1\mod q$ for every nontrivial divisor $p$ of $q$, and we can compute $g^p\bmod q$ in $O(\log q)$ time via repeated squaring). Then, in $O(q)$ time create a lookup table $A[0\ldots q-1]$ with $A[i] := g^i\mod q$. Then create an inverse lookup table by for each $0\le i < q$, setting the inverse of $A[i]$ to $A[q-1-i]$.
\end{proof}

\section{\hybridprojectivegeometry: trading off error and time}\label{sec:tradeoff}

In this section, we describe a hybrid scheme using an intermediate
value for the field size $q$ to trade off between the variance and
the running time. Roughly speaking, larger values for $q$ lead to
slower running time but also smaller variance.
The approach is similar to the way~\cite{AcharyaSZ19}
extended their scheme from the high privacy regime to the general setting.
We choose $h,q,t$ such that they satisfy the following
conditions:
\begin{itemize}
\item $b=\frac{q^{t}-1}{q-1}$ and $bh\ge k> c_{set}h$. 
\item Let $c_{set}=\frac{q^{t-1}-1}{q-1}$, $c_{int}=\frac{q^{t-2}-1}{q-1}$,
and $z=c_{set}/c_{int}$. Note that $c_{set}^{2}\ge b\cdot c_{int}$
and $q+1\ge z\ge q$.
\item Choose $hz$ as close as possible to $e^{\varepsilon}+1$.
\end{itemize}
The input coordinates are partitioned into blocks of size at most
$b$ each. The algorithm's response consists of two parts: the index
of the block and the index inside the block. First, the algorithm
uses the randomized response to report the block. Next, if the response
has the correct block then the algorithm uses the scheme described
in the previous section with field size $q$ to describe the coordinate
inside the block. If the first response has the wrong block then the
algorithm uses a uniformly random response in the second part.

More precisely, the algorithm works as follows. Each input value is
identified with a pair $(i,v)$ where $i\in\Z_{h}$ and $v$ is a canonical
vector in $\F_{q}^{t}$. If $k<bh$ then we allocate up to $\lceil k/h\rceil$
input values to each block.
The response is a pair $(j,u)$ where $i\in\Z_{h}$
and $u$ is a canonical vector in $\F_{q}^{t}$ chosen as follows.
For $j=i$ and $\left\langle u,v\right\rangle =0$, the pair $(j,u)$
is chosen with probability $e^{\varepsilon}p$. All other choices are
chosen with probability $p$ each. Because all probabilities are either
$p$ or $e^{\varepsilon}p$, the scheme is $\varepsilon$-private. We have

\begin{align*}
e^{\varepsilon}p\cdot c_{set}+p\left(bh-c_{set}\right) & =1\\
p & =\frac{1}{bh+\left(e^{\varepsilon}-1\right)c_{set}}
\end{align*}

Let $\tilde{x}_{i,v}$ be our estimate for the frequency of input
$(i,v)$. The estimates are computed as follows.

\[
\tilde{x}_{i,v}=\alpha\left(\sum_{\left\langle v,u\right\rangle =0}y_{i,u}\right)+\beta\left(\sum_{u}y_{i,u}\right)+\gamma\left(\sum_{j,u}y_{j,u}\right)
\]

We need to choose $\alpha,\beta$ and $\gamma$ so that $\tilde{x}_{i,v}$
is an unbiased estimator of $x_{i,v}$. By linearity of expectation,
we only need to consider the case with exactly one user. If the input
is $i,v$ then we have
\[
\Ex\left[\tilde x_{i,v}\right]=\alpha e^{\varepsilon}pc_{set}+\beta p\left(\left(e^{\varepsilon}-1\right)c_{set}+b\right)+\gamma=1
\]

If the input is not $i,v$ but in the same block then

\[
\Ex\left[\tilde x_{i,v}\right]=\alpha p\left(\left(e^{\varepsilon}-1\right)c_{int}+c_{set}\right)+\beta p\left(\left(e^{\varepsilon}-1\right)c_{set}+b\right)+\gamma=0
\]

Finally if the input is in a different block then
\[
\Ex\left[\tilde x_{i,v}\right]=\alpha pc_{set}+\beta pb+\gamma=0
\]

We solve for $\alpha,\beta,\gamma$ and get

\begin{align*}
\alpha & =\frac{1}{p\left(e^{\varepsilon}-1\right)\left(c_{set}-c_{int}\right)}=\frac{bh+\left(e^{\varepsilon}-1\right)c_{set}}{\left(e^{\varepsilon}-1\right)\left(c_{set}-c_{int}\right)}\\
\beta & =-\frac{\alpha c_{int}}{c_{set}}=-\frac{c_{int}/c_{set}}{p\left(e^{\varepsilon}-1\right)\left(c_{set}-c_{int}\right)}\\
&=-\frac{bh+\left(e^{\varepsilon}-1\right)c_{set}}{\left(e^{\varepsilon}-1\right)\left(c_{set}-c_{int}\right)}\cdot\frac{c_{int}}{c_{set}}\\
\gamma & =-\alpha pc_{set}-\beta pb=-\frac{c_{set}-b\cdot\frac{c_{int}}{c_{set}}}{\left(e^{\varepsilon}-1\right)\left(c_{set}-c_{int}\right)}\le0
\end{align*}

We note that $\alpha+\beta=\left(1-\frac{c_{int}}{c_{set}}\right)\alpha=\frac{bh/c_{set}+\left(e^{\varepsilon}-1\right)}{\left(e^{\varepsilon}-1\right)}$.

\begin{lemma}\label{lem:tradeoff-error}
\begin{align*}
&\Ex\left[\left\Vert x-\tilde{x}\right\Vert_2^{2}\right]\le n\left(1+\frac{\left(zh+\left(e^{\varepsilon}-1\right)\right)}{\left(e^{\varepsilon}-1\right)^{2}\left(z-1\right)}+\frac{2}{\left(e^{\varepsilon}-1\right)}+\frac{e^{\varepsilon}\left(zh-e^{\varepsilon}+1\right)}{\left(e^{\varepsilon}-1\right)^{2}}\right)\\
&+n\frac{\left(zh+\left(e^{\varepsilon}-1\right)\right)z}{\left(e^{\varepsilon}-1\right)^{2}\left(z-1\right)}\left(k-\left\lceil k/h\right\rceil +(\left\lceil k/h\right\rceil -1)\frac{\left(z+e^{\varepsilon}-1\right)}{z}\right)
\end{align*}
In particular, if $zh=e^{\varepsilon}+1$ then $\Ex\left[\frac 1k\left\Vert x-\tilde{x}\right\Vert_2^{2}\right]\le \frac nk+\frac{z}{z-1}\cdot n\frac{4e^{\varepsilon}}{\left(e^{\varepsilon}-1\right)^{2}}$
\end{lemma}

\begin{proof}
By independence, we only need to analyze the variance when there is
exactly one user with input $(i,v)$ and response $(j,u)$. The lemma follows from adding up the variances from all users.

\begin{align*}
\Ex\left[\left(\tilde{x}_{i,v}-1\right)^{2}\right] \le &\Ex\left[\left(\tilde{x}_{i,v}-1-\gamma\right)^{2}\right]\\
 =&\Pr\left[j\ne i\right]\cdot\left(-1\right)^{2}+\Pr\left[j=i\wedge\left\langle u,v\right\rangle \ne0\right]\left(\beta-1\right)^{2}
 +\Pr\left[j=i\wedge\left\langle u,v\right\rangle =0\right]\left(\alpha+\beta-1\right)^{2}\\
 =&\left(1-\left(e^{\varepsilon}-1\right)pc_{set}-pb\right)+p\left(b-c_{set}\right)\left(\beta-1\right)^{2}
 +e^{\varepsilon}pc_{set}\left(\alpha+\beta-1\right)^{2}\\
 =&1+p\left(b-c_{set}\right)\left(\beta^2-2\beta\right)
 +e^{\varepsilon}pc_{set}\left(\alpha+\beta\right)\left(\alpha+\beta-2\right)
\end{align*}
We expand the second and third terms individually:
\begin{align*}
    &p\left(b-c_{set}\right)\left(\beta^2-2\beta\right)\\
    &=p\left(b-c_{set}\right)\left(\left(\frac{c_{int}/c_{set}}{p\left(e^{\varepsilon}-1\right)\left(c_{set}-c_{int}\right)}\right)^{2}+\frac{2c_{int}/c_{set}}{p\left(e^{\varepsilon}-1\right)\left(c_{set}-c_{int}\right)}\right)\\
    &=\frac{\left(bh+\left(e^{\varepsilon}-1\right)c_{set}\right)\left(b-c_{set}\right)c_{int}^{2}/c_{set}^{2}}{\left(e^{\varepsilon}-1\right)^{2}\left(c_{set}-c_{int}\right)^{2}}+\frac{2\left(b-c_{set}\right)c_{int}/c_{set}}{\left(e^{\varepsilon}-1\right)\left(c_{set}-c_{int}\right)}\\
    &=\frac{\left(bh/c_{set}+\left(e^{\varepsilon}-1\right)\right)\left(b/c_{set}-1\right)}{\left(e^{\varepsilon}-1\right)^{2}\left(c_{set}/c_{int}-1\right)^{2}}+\frac{2\left(b/c_{set}-1\right)}{\left(e^{\varepsilon}-1\right)\left(c_{set}/c_{int}-1\right)}\\
    &\le\frac{\left(zh+\left(e^{\varepsilon}-1\right)\right)}{\left(e^{\varepsilon}-1\right)^{2}\left(z-1\right)}+\frac{2}{\left(e^{\varepsilon}-1\right)}
\end{align*}
and
\begin{align*}
    &e^{\varepsilon}pc_{set}\left(\alpha+\beta\right)\left(\alpha+\beta-2\right)\\
    &=\frac{e^{\varepsilon}c_{set}}{bh+\left(e^{\varepsilon}-1\right)c_{set}}\cdot\frac{bh/c_{set}+\left(e^{\varepsilon}-1\right)}{\left(e^{\varepsilon}-1\right)}\cdot \frac{bh/c_{set}-\left(e^{\varepsilon}-1\right)}{\left(e^{\varepsilon}-1\right)}\\
    &=\frac{e^{\varepsilon}\left(bh/c_{set}-e^{\varepsilon}+1\right)}{\left(e^{\varepsilon}-1\right)^{2}}\\
    &\le\frac{e^{\varepsilon}\left(zh-e^{\varepsilon}+1\right)}{\left(e^{\varepsilon}-1\right)^{2}}
\end{align*}

When $zh=e^{\varepsilon}+1$, we have 
\begin{align*}
\Ex\left[\left(\tilde{x}_{i,v}-1\right)^{2}\right]\le
1+\frac{2e^{\varepsilon}}{\left(e^{\varepsilon}-1\right)^{2}\left(z-1\right)}+\frac{2}{\left(e^{\varepsilon}-1\right)}+\frac{2e^{\varepsilon}}{\left(e^{\varepsilon}-1\right)^{2}}<1+\frac{4e^{\varepsilon}}{\left(e^{\varepsilon}-1\right)^{2}}\frac{z}{z-1}
\end{align*}
Next consider $v'\ne v$.

\begin{align*}
\Ex\left[\left(\tilde{x}_{i,v'}-0\right)^{2}\right] & \le\Ex\left[\left(\tilde{x}_{i,v'}-\gamma\right)^{2}\right]\\
 & =\Pr\left[j\ne i\right]\cdot0+\Pr\left[j=i\wedge\left\langle u,v'\right\rangle \ne0\right]\left(\beta\right)^{2}+\Pr\left[j=i\wedge\left\langle u,v'\right\rangle =0\right]\left(\alpha+\beta\right)^{2}\\
 & =p\left(b+\left(e^{\varepsilon}-2\right)c_{set}-\left(e^{\varepsilon}-1\right)c_{int}\right)\left(\beta\right)^{2}+p\left(\left(e^{\varepsilon}-1\right)c_{int}+c_{set}\right)\left(\alpha+\beta\right)^{2}\\
 & =\left(b+\left(e^{\varepsilon}-2\right)c_{set}-\left(e^{\varepsilon}-1\right)c_{int}\right)\frac{1}{p\left(e^{\varepsilon}-1\right)^{2}\left(c_{set}-c_{int}\right)^{2}}\frac{c_{int}^{2}}{c_{set}^{2}}+\\
 & \ \ \ p\left(\left(e^{\varepsilon}-1\right)c_{int}+c_{set}\right)\frac{\left(1-c_{int}/c_{set}\right)^{2}}{p^{2}\left(e^{\varepsilon}-1\right)^{2}\left(c_{set}-c_{int}\right)^{2}}\\
 & =\left(b/c_{set}+\left(e^{\varepsilon}-2\right)-\left(e^{\varepsilon}-1\right)c_{int}/c_{set}\right)\frac{\left(bh/c_{set}+\left(e^{\varepsilon}-1\right)\right)}{\left(e^{\varepsilon}-1\right)^{2}\left(1-c_{int}/c_{set}\right)^{2}}\frac{c_{int}^{2}}{c_{set}^{2}}+\\
 & \ \ \ \left(\left(e^{\varepsilon}-1\right)c_{int}/c_{set}+1\right)\frac{\left(bh/c_{set}+\left(e^{\varepsilon}-1\right)\right)}{\left(e^{\varepsilon}-1\right)^{2}}\\
 & \le\left(z+\left(e^{\varepsilon}-2\right)-\left(e^{\varepsilon}-1\right)/z\right)\frac{\left(zh+\left(e^{\varepsilon}-1\right)\right)}{\left(e^{\varepsilon}-1\right)^{2}\left(z-1\right)^{2}}+\left(\left(e^{\varepsilon}-1\right)/z+1\right)\frac{\left(zh+\left(e^{\varepsilon}-1\right)\right)}{\left(e^{\varepsilon}-1\right)^{2}}\\
 & =\left(z+e^{\varepsilon}-1\right)\frac{\left(zh+\left(e^{\varepsilon}-1\right)\right)}{\left(e^{\varepsilon}-1\right)^{2}\left(z-1\right)}
\end{align*}

When $zh=e^{\varepsilon}+1$, the last expression is bounded by $\left(z+e^{\varepsilon}-1\right)\frac{2e^{\varepsilon}}{\left(e^{\varepsilon}-1\right)^{2}\left(z-1\right)z}+\left(z+e^{\varepsilon}-1\right)\frac{2e^{\varepsilon}}{\left(e^{\varepsilon}-1\right)^{2}z}=\frac{\left(z+e^{\varepsilon}-1\right)}{\left(z-1\right)}\cdot\frac{2e^{\varepsilon}}{\left(e^{\varepsilon}-1\right)^{2}}\le\frac{1+h}{z}\frac{z}{\left(z-1\right)}\cdot\frac{2e^{\varepsilon}}{\left(e^{\varepsilon}-1\right)^{2}}$

Finally, consider $i'\ne i$ and arbitrary $v'$.

\begin{align*}
\Ex\left[\left(\tilde{x}_{i',v'}-0\right)^{2}\right] & \le\Ex\left[\left(\tilde{x}_{i',v'}-\gamma\right)^{2}\right]\\
 & =\Pr\left[j\ne i'\right]\cdot0^{2}+\Pr\left[j=i'\wedge\left\langle u,v'\right\rangle \ne0\right]\left(\beta\right)^{2}+\Pr\left[j=i'\wedge\left\langle u,v'\right\rangle =0\right]\left(\alpha+\beta\right)^{2}\\
 & =p\left(b-c_{set}\right)\left(\beta\right)^{2}+pc_{set}\left(\alpha+\beta\right)^{2}\\
 & =p\left(b-c_{set}\right)\frac{1}{p^{2}\left(e^{\varepsilon}-1\right)^{2}\left(c_{set}-c_{int}\right)^{2}}\frac{c_{int}^{2}}{c_{set}^{2}}+pc_{set}\frac{\left(1-c_{int}/c_{set}\right)^{2}}{p^{2}\left(e^{\varepsilon}-1\right)^{2}\left(c_{set}-c_{int}\right)^{2}}\\
 & =\left(b/c_{set}-1\right)\frac{\left(bh/c_{set}+\left(e^{\varepsilon}-1\right)\right)}{\left(e^{\varepsilon}-1\right)^{2}\left(1-c_{int}/c_{set}\right)^{2}}\frac{c_{int}^{2}}{c_{set}^{2}}+\frac{bh/c_{set}+\left(e^{\varepsilon}-1\right)}{\left(e^{\varepsilon}-1\right)^{2}}\\
 & \le\frac{\left(zh+\left(e^{\varepsilon}-1\right)\right)}{\left(e^{\varepsilon}-1\right)^{2}\left(z-1\right)}+\frac{zh+\left(e^{\varepsilon}-1\right)}{\left(e^{\varepsilon}-1\right)^{2}}\\
 & =\frac{\left(zh+\left(e^{\varepsilon}-1\right)\right)z}{\left(e^{\varepsilon}-1\right)^{2}\left(z-1\right)}
\end{align*}

When $zh=e^{\varepsilon}+1$, the last expression is bounded by $\frac{2e^{\varepsilon}}{\left(e^{\varepsilon}-1\right)^{2}}\frac{1}{z-1}+\frac{2e^{\varepsilon}}{\left(e^{\varepsilon}-1\right)^{2}}=\frac{2e^{\varepsilon}}{\left(e^{\varepsilon}-1\right)^{2}}\frac{z}{z-1}$

There are $b_{i}\le\left\lceil k/h\right\rceil \le b$ valid coordinates
in the same block with the input $(i,v)$. There are $k-b_{i}$ coordinates
in the other blocks. Thus the total variance across all coordinates except for coordinate $(i,v)$
is bounded by

\begin{align*}
\frac{\left(zh+\left(e^{\varepsilon}-1\right)\right)z}{\left(e^{\varepsilon}-1\right)^{2}\left(z-1\right)}\left(k-b_{i}+(b_{i}-1)\frac{\left(z+e^{\varepsilon}-1\right)}{z}\right)\\
\le\frac{\left(zh+\left(e^{\varepsilon}-1\right)\right)z}{\left(e^{\varepsilon}-1\right)^{2}\left(z-1\right)}\left(k-\left\lceil k/h\right\rceil +(\left\lceil k/h\right\rceil -1)\frac{\left(z+e^{\varepsilon}-1\right)}{z}\right)
\end{align*}

For $zh=e^{\varepsilon}+1$, we have $\frac{\left(z+e^{\varepsilon}-1\right)}{z}\le1+h$
and $k-b_{i}+(b_{i}-1)\frac{\left(z+e^{\varepsilon}-1\right)}{z}\le k-\left\lceil k/h\right\rceil +(\left\lceil k/h\right\rceil -1)(1+h)=k-\left\lceil k/h\right\rceil +\left\lceil k/h\right\rceil -1+h(\left\lceil k/h\right\rceil -1)<2k$.
\end{proof}
Regarding the decoding algorithms, notice that the estimates are computed
separately by blocks except for an offset $\gamma$ scaled by the
total number of received messages across all blocks. Thus, using the
naive algorithm, the time to estimate one count is $O\left(c_{set}\right)=O\left(q^{\left\lceil \log_{q}\left(k/h\right)\right\rceil -1}\right)$.
Using the fast algorithm to estimate all counts takes $O\left(bqt\right)$
time per block and in total, $O\left(bqth\right)=O\left(\left\lceil \log_{q}\left(k/h\right)\right\rceil hq^{1+\left\lceil \log_{q}\left(k/h\right)\right\rceil }\right)$
time.

\section{Experimental Results}
\label{sec:experiments}
In this section, we compare previously-known algorithms (\rappor, \pirappor, \hadamardresponse (\hr), \recursivehadamardresponse (\rhr), \subsetselection (\SS)) and our new algorithms \projectivegeometry (\pg) and \hybridprojectivegeometry (\hpg). As the variance upper bound of these algorithms do not depend on the underlying data, we perform our experiments on simple synthetic data that realize the worst case for variance. Our experiments show that \textsf{ProjectiveGeometryOracle} matches the best of these algorithms namely \SS, \rappor, and \pirappor, and achieves noticeably better MSE than other communication- and computation-efficient approaches. At the same time it is significantly more efficient than those three in terms of server computation time, while also achieving optimal communication.

\begin{table*}
  \begin{center}
  \begin{tabular}{|c|c|}
    \hline
    \textbf{scheme name} & \textbf{runtime (in seconds)}\\
    \hline
    \pirappor & 1{,}893.82 (approximately 31.5 minutes)\\
    \hline
    \pg & 36.92 \\
    \hline
    \hpg & 5.94\\
    \hline
    \rhr & 1.20\\
    \hline
    \hr & 0.64\\
    \hline
    \rr & 0.02\\
    \hline
  \end{tabular}
  \caption{Server runtimes for $\eps = 5$, $k = 3{,}307{,}948$. For \hpg, we chose the parameters $h = 50, q = 3, t = 11$, so that the mechanism rounded up the universe size to $h(q^t-1)/(q-1)$, which is about 34\% larger than $k$.}\label{fig:runtimes}
  \end{center}
\end{table*}

All experiments were run on a Dell Precision T3600 with six Intel 3.2 GHz Xeon E5-1650 cores running Ubuntu 20.04 LTS, though our implementation did not take advantage of parallelism. We implemented all algorithms and ran experiments in C++, using the GNU C++ compiler version 9.3.0; code and details on how to run the experiments used to generate data and plots are in our public repository at \url{https://github.com/minilek/private_frequency_oracles/}.

We first performed one experiment to show the big gap in running times. We took $\eps = 5$, a practically relevant setting, and $n = 10{,}000$, $k = 3{,}307{,}948$; this setting of $n$ is smaller than one would see in practice, but the runtimes of the algorithms considered are all linear in $n$ plus additional terms that depend on $k,\eps$, and our aim was to measure the impact of these additive terms, which can be significant even for large $n$. Furthermore, in practice the server can immediately process messages from each of the $n$ users dynamically as the messages arrive asynchronously, whereas it is the additive terms that must be paid at once at the time of histogram reconstruction. For our settings, the closest prime to $\exp(\eps)+1 \approx 149.4$ is $q=151$. Recall that \pg rounds up to universe sizes of the form $(q^t-1)/(q-1)$; then $(q^4-1)/(q-1)$ is less than 5\% larger than $k$, so that the negative effect of such rounding on the runtime of \pg is minimal. Meanwhile \pirappor picks the largest prime $q'$ smaller than $\exp(\eps)+1$, which in this case is $q' = 149$, and assumes universe sizes of the form $q'^t - 1$; in this case $q'^3 - 1 = k$ exactly, so rounding issues do not negatively impact the running time of \pirappor (we chose this particular value of $k$ intentionally, to show \pirappor's performance in the best light for some fairly large universe size). The runtimes of various algorithms with this setting of $\eps,k,n$ are shown in \cref{fig:runtimes}. Note \rhr and \hr sacrifice a constant factor in utility compared to \pirappor and \pg, the former of which is four orders of magnitude slower while the latter is only one order of magnitude slower and approximately 51x faster than \pirappor. Meanwhile, \hpg's runtime is of the same order of magnitude (though roughly 5x slower) than \rhr, but as we will see shortly, \hpg can provide significantly improved utility over \rhr and \hr.


\begin{figure}
     \centering
     \begin{subfigure}[b]{0.48\textwidth}
         \centering
         \includegraphics[width=\textwidth]{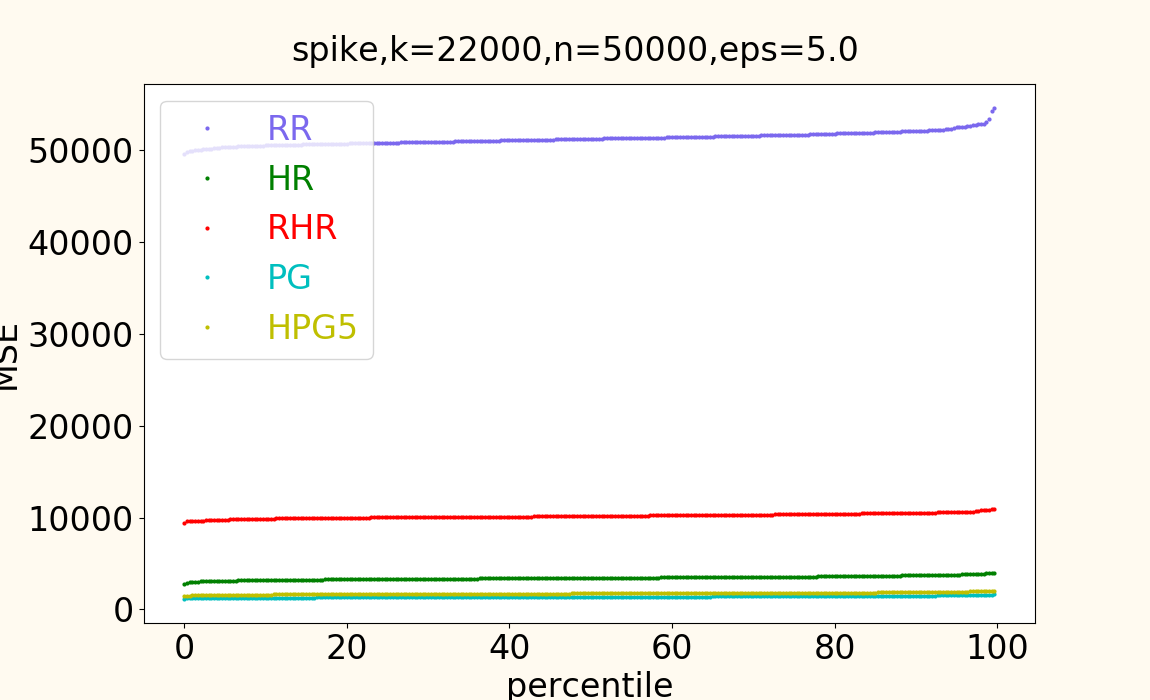}
         \caption{}
         \label{fig:spike5mse}
     \end{subfigure}
     \hfill
     \begin{subfigure}[b]{0.48\textwidth}
         \centering
         \includegraphics[width=\textwidth]{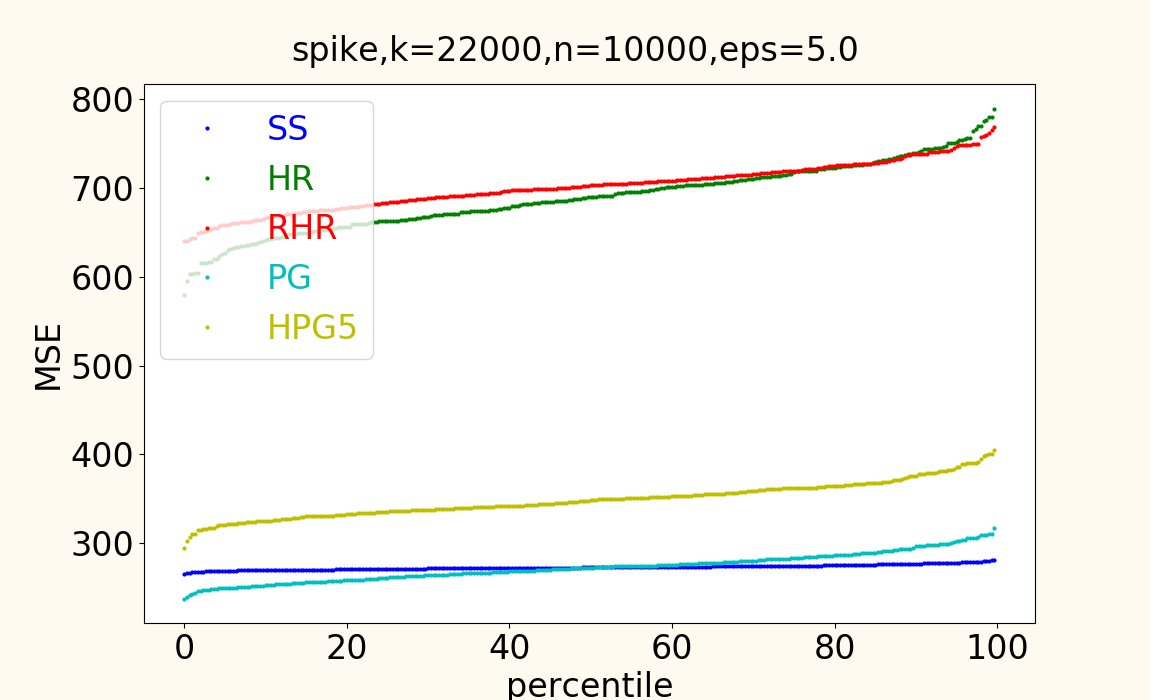}
         \caption{}
         \label{fig:spike7mse}
     \end{subfigure}
        \caption{\randomizedresponse has significantly worse error than other algorithms, even for moderately large universes, followed by \hadamardresponse and \recursivehadamardresponse, which have roughly double the error of state-of-the-art algorithms. \hybridprojectivegeometry trades off having slightly worse error than state-of-the-art for faster runtime.}
        \label{fig:rrbad}
\end{figure}

\begin{figure*}[!h]
     \centering
     \begin{subfigure}[b]{0.45\textwidth}
         \centering
         \includegraphics[width=\textwidth]{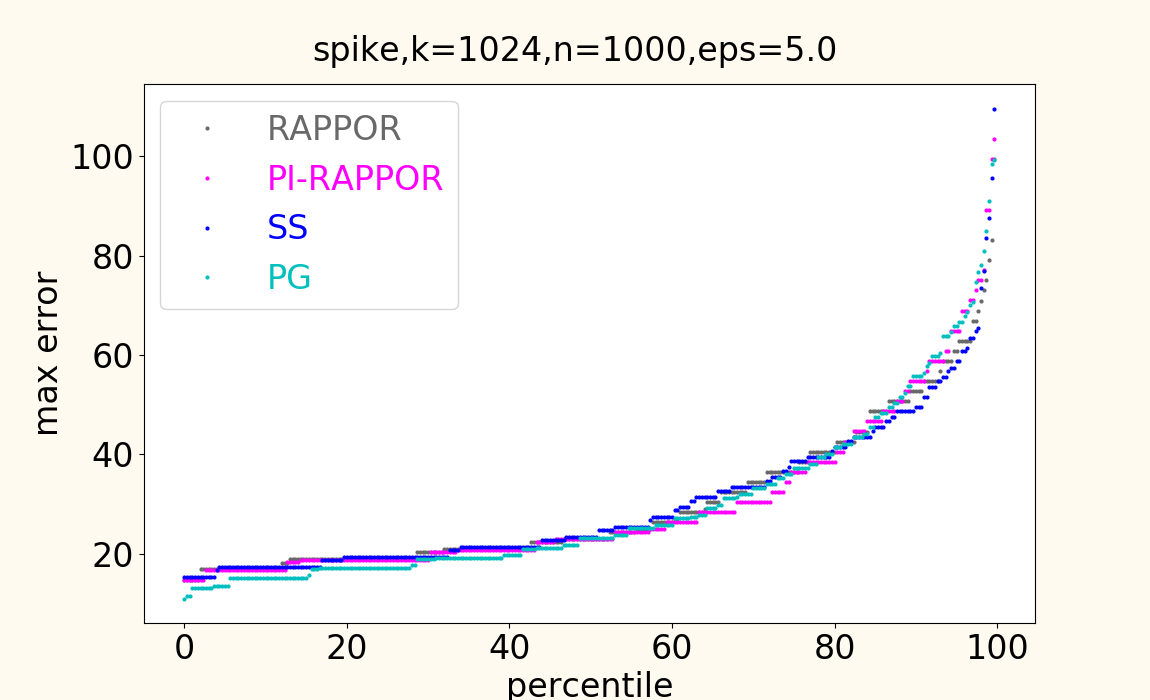}
         \caption{}
         \label{fig:spike1mse}
     \end{subfigure}
     \hfill
     \begin{subfigure}[b]{0.45\textwidth}
         \centering
         \includegraphics[width=\textwidth]{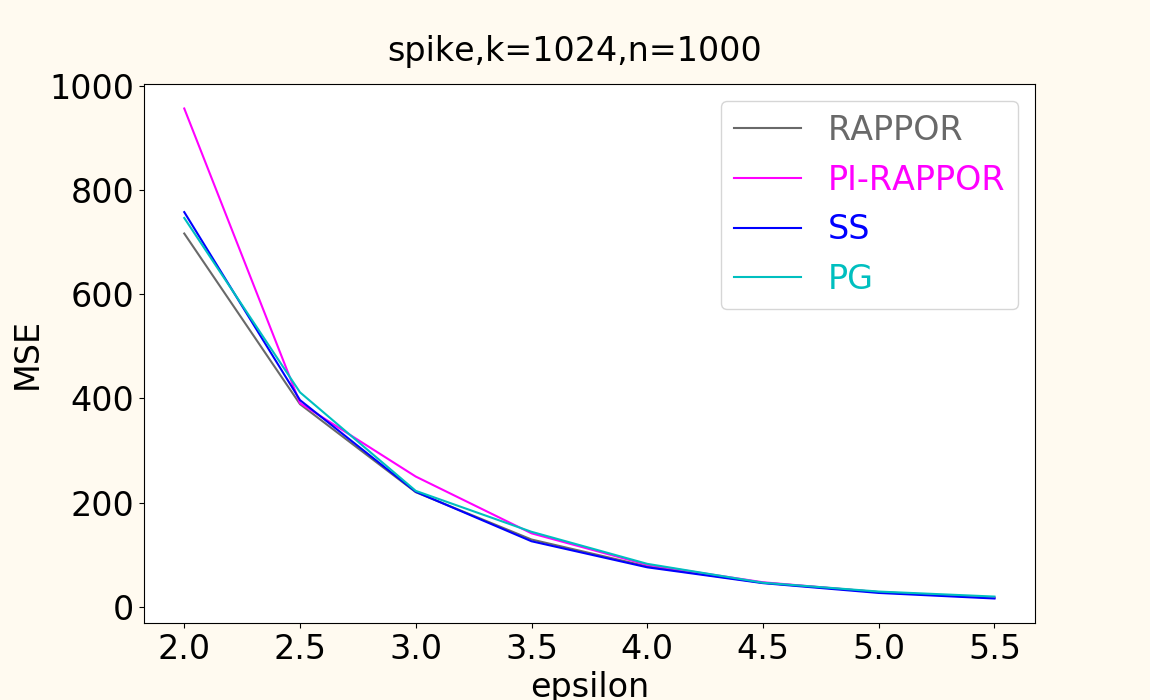}
         \caption{}
         \label{fig:spike1var}
     \end{subfigure}
     \hfill
     \begin{subfigure}[b]{0.45\textwidth}
         \centering
         \includegraphics[width=\textwidth]{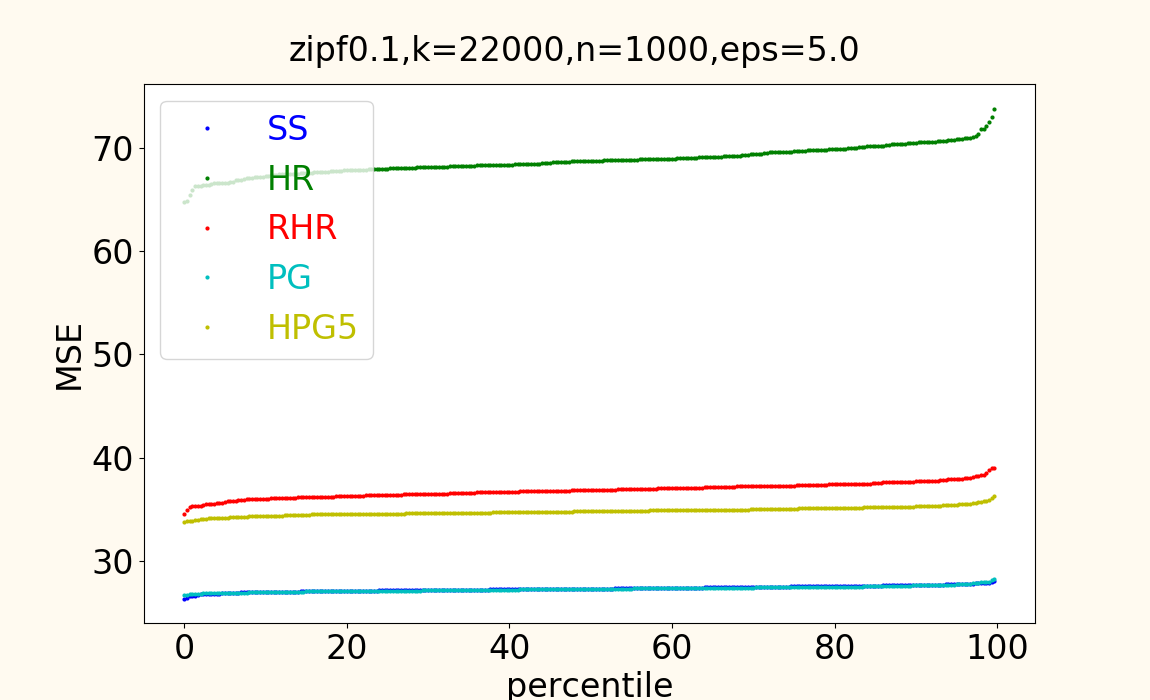}
         \caption{}
         \label{fig:zipf0.1mse}
 \end{subfigure}
      \hfill
     \begin{subfigure}[b]{0.45\textwidth}
         \centering
         \includegraphics[width=\textwidth]{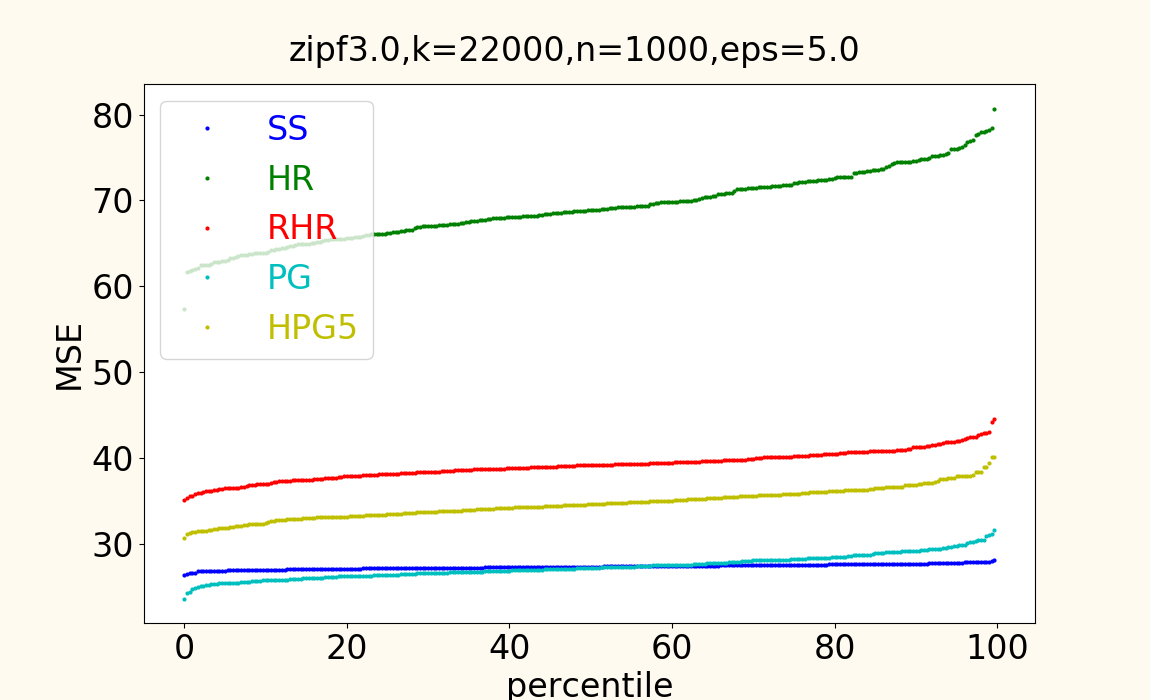}
         \caption{}
         \label{fig:zipf3.0mse}
     \end{subfigure}
     \hfill
     \begin{subfigure}[b]{0.45\textwidth}
         \centering
         \includegraphics[width=\textwidth]{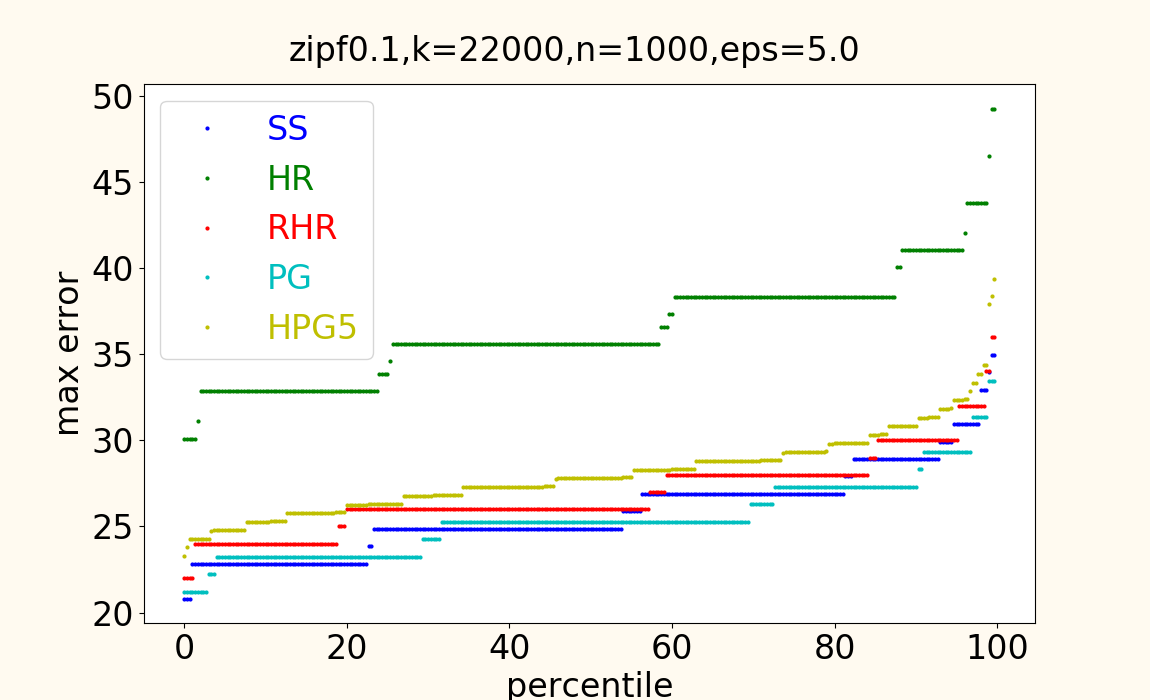}
         \caption{}
         \label{fig:zipf0.1max}
     \end{subfigure}
     \hfill
     \begin{subfigure}[b]{0.45\textwidth}
         \centering
         \includegraphics[width=\textwidth]{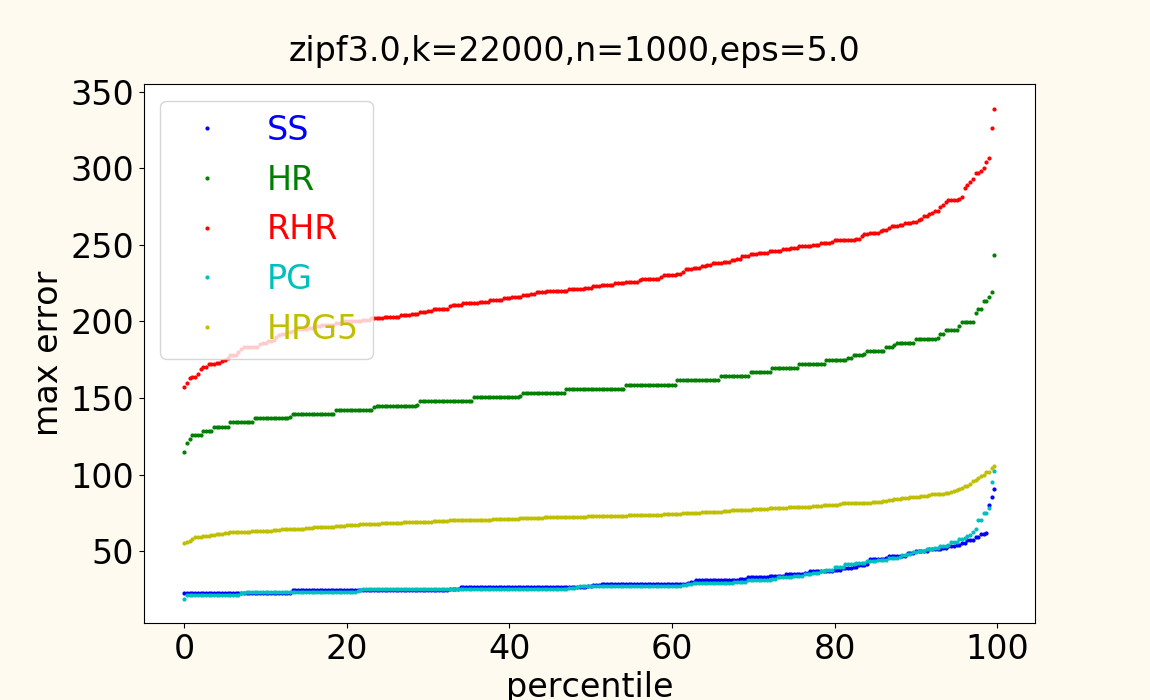}
         \caption{}
         \label{fig:zipf3.0max}
     \end{subfigure}
     \hfill
     \begin{subfigure}[b]{0.45\textwidth}
         \centering
         \includegraphics[width=\textwidth]{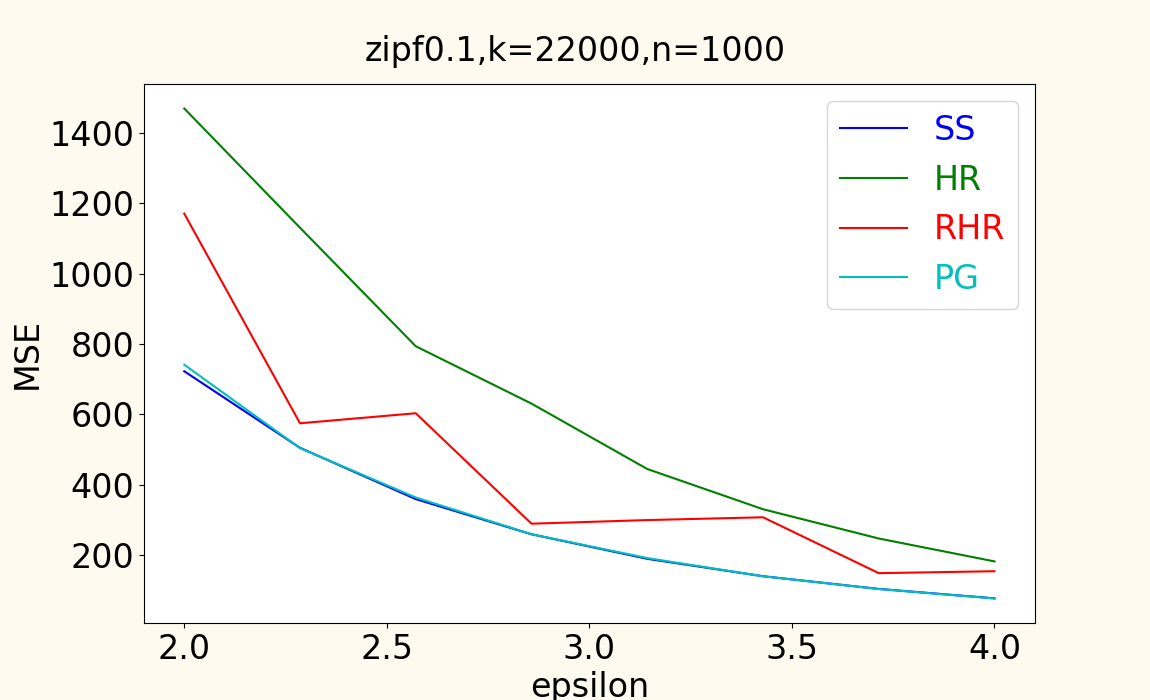}
         \caption{}
         \label{fig:zipf0.1var}
     \end{subfigure}
     \begin{subfigure}[b]{0.45\textwidth}
         \centering
         \includegraphics[width=\textwidth]{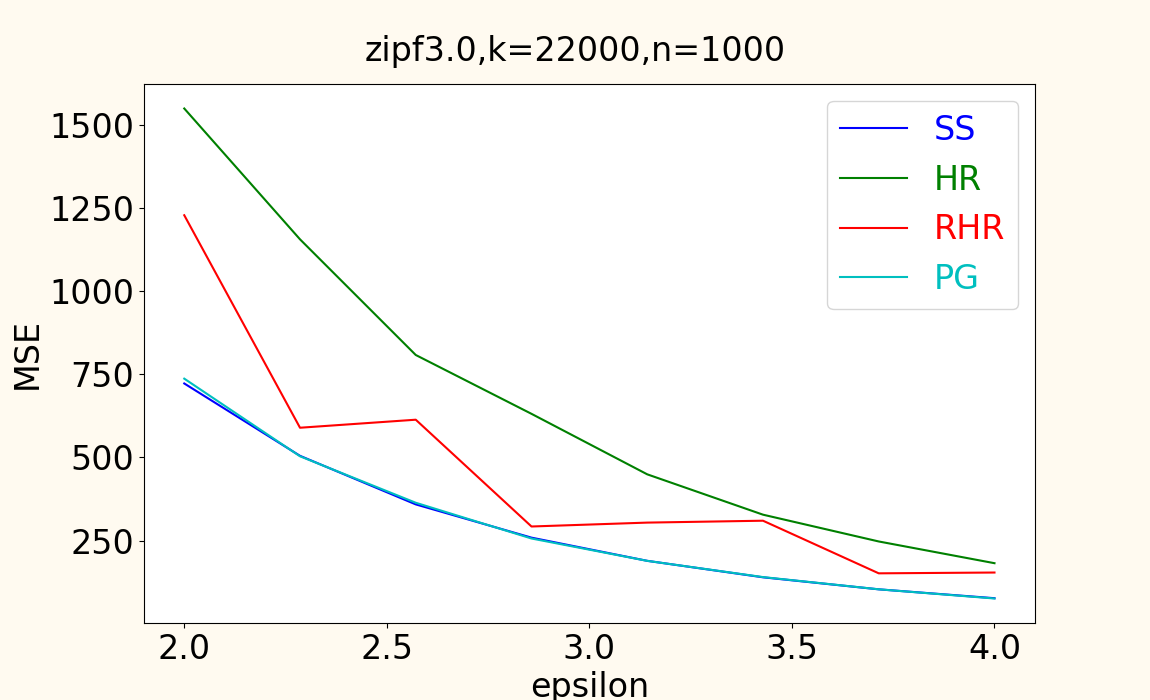}
         \caption{}
         \label{fig:zipf3.0var}
     \end{subfigure}
        \caption{Error distributions from experiments.}
        \label{fig:experiments}
\end{figure*}

Next we discuss error. Many of our experiments showing reconstruction error with fixed $\eps$ take $\eps=5$, a practically relevant setting, and universe size $k = 22{,}000$, for which the closest prime to $\exp(\eps)+1\approx 149.4$ is $q=151$. Recall that \pg rounds up to universe sizes of the form $(q^t-1)/(q-1)$; then $(q^3-1)/(q-1) = 22{,}593$ is not much larger than $k$, so that the runtime of \pg is not severely impacted. Also, $c_{set}/c_{int}$ as defined in \cref{sec:basic-construction} is very close to $\exp(\eps)+1$, so that the MSE bound in \cref{lem:pg-var} nearly matches that of \SS. Furthermore for \hpg for this setting of $\eps, k$, if we choose $q=5, h = 30, t = 5$, then $h\cdot (q^t-1)/(q-1) = 23{,}430$, which is not much bigger than $k$ so that the runtime of \hpg is not majorly impacted. Furthermore $hz$ as defined in \cref{sec:tradeoff} is approximately $150.19$, which is very close to $\exp(\eps)+1$ as recommended by \cref{lem:tradeoff-error} to obtain minimal error. We first draw attention to \cref{fig:spike1mse,fig:spike1var}. These plots run \rappor, \pirappor, \pg, and \SS with $k,n,\eps$ as in the figure and show that their error distributions are essentially equivalent. We show the plots for only one some particular parameter settings, but the picture has looked essentially the same to us regardless of which parameters we have tried. In \cref{fig:spike1mse}, we have $n$ users each holding the same item in the universe (item $0$); we call this a {\it spike} distribution as noted in the plot. We have each user apply its local randomizer to send a message to the server, and we ask the server to then reconstruct the histogram (which should be $(n,0,\ldots,0)$) and calculate the MSE. We repeat this experiment 300 times, and in this plot we have 300 dots plotted per algorithm, where a dot at point $(x,y)$ signifies that the MSE was at most $y$ for $x\%$ of the trial runs; this, it is a plot of the CDF of the empirical error distribution. In \cref{fig:spike1var}, we plot MSE as a function of increasing $\eps$, where for each value of $\eps$ we repeat the above experiment 10 times then plot the average MSE across those 10 trial runs. Because the error performance of \rappor, \pirappor, \SS, and \pg are so similar, in all other plots we do not include \rappor and \pirappor since their runtimes are so slow that doing extensive experiments is very time-consuming computationally (note: our implementation of \rappor requires $O(nk)$ server time, though $O(n(k/e^\eps + 1))$ expected time is possible by having each user transmit only a sparse encoding of the locations of the $1$ bits in its message). We finally draw attention to \cref{fig:zipf0.1mse,fig:zipf3.0mse,fig:zipf0.1max,fig:zipf3.0max,fig:zipf0.1var,fig:zipf3.0var}. Here we run several algorithms where the distribution over the universe amongst the users is Zipfian (a power law), with power law exponent either $0.1$ (an almost flat distribution), or $3.0$ (rapid decay). The $\hpg$ algorithm was run with $q=5$. As can be seen, the qualitative behavior and relative ordering of all the algorithms is essentially unchanged by the Zipf parameter: \pg,\SS always have the best error, followed by \hpg, followed by \rhr and \hr. \cref{fig:zipf0.1mse,fig:zipf3.0mse} show the CDF of the empirical MSE over $300$ independent trials, as discussed above. \cref{fig:zipf0.1max,fig:zipf3.0max} is similar, but the $y$-axis denotes $\|x - \tilde x\|_\infty$ instead of the MSE. \cref{fig:zipf0.1var,fig:zipf3.0var} shows how the MSE varies as $\eps$ is increased; in these last plots we do not include \hpg as one essentially one should select a different $q$ for each $\eps$ carefully to obtain a good tradeoff between runtime and error (as specified by \cref{lem:tradeoff-error}) due to round-up issues in powering $q$. 

\section*{Acknowledgments}
We thank Noga Alon for pointing out the relevance of projective geometry for constructing the type of set system our mechanism relies on.

\newcommand{\etalchar}[1]{$^{#1}$}

\appendix

\section{Fast dynamic programming for PI-RAPPOR}
In this section, we describe an adaptation of our dynamic programming approach to PI-RAPPOR. First, we briefly review the construction of PI-RAPPOR.
We use $\mathbb{F}_{q}$ with the field size $q$ close to $e^{\eps}+1$. Let $t$ be the minimum integer such that $k\le q^{t}$.

We identify the $k$ input values with vectors in $\mathbb{F}_{q}^{t}$. Let 
$x\in \Z^{q^t}$ denote the input frequency vector i.e. $x_{v}$ is the number
of users with input $v\in\mathbb{F}_{q}^{t}$. For each input $v$,
we define a set $S(v)\subset\mathbb{F}_{q}^{t}\times\mathbb{F}_{q}$
where $(a,b)\in S(v)$ if and only if $\left\langle a,v\right\rangle +b=0$.

Each user with input $v$ sends a random element $e$ of $\mathbb{F}_{q}^{t}\times\mathbb{F}_{q}$
with probability $e^{\eps}p$ if $e\in S(v)$ and probability $p$
if $e\not\in S(v)$. Thus, $p=\frac{1}{e^\eps q^t + (q-1)q^t}$. The server keeps
the counts on the received elements in a vector $y$ indexed
by elements of $\mathbb{F}_{q}^{t}\times\mathbb{F}_{q}$. The total
storage is $O\left(q^{t+1}\right)$. We estimate the frequency vector $x$ by computing

\[
\tilde{x}_{v}=\alpha\left(\sum_{u,w:\left\langle u,v\right\rangle +w=0}y_{u,w}\right)+\beta\sum_{u,w}y_{u,w}
\]
where $\alpha$ and $\beta$ are chosen so that this is an unbiased estimator. This condition implies two equations:
\begin{align*}
    \alpha \frac{e^\eps q^t}{e^\eps q^t + (q-1)q^t} + \beta &= 1\\
    \alpha \frac{e^\eps q^{t-1}+(q-1)q^{t-1}}{e^\eps q^t + (q-1)q^t} + \beta &= 0
\end{align*}
We obtain
\begin{align*}
    \alpha &= \frac{e^\eps q + (q-1)q}{(e^\eps-1)(q-1)}\\
    \beta  &= -\frac{e^\eps+(q-1)}{(e^\eps-1)(q-1)}
\end{align*}

Next, we describe a fast algorithm to compute $\tilde{x}$ with running
time $O\left(tq^{t+2}\right)$. Specifically, for $a\in\mathbb{F}_{q}^{j},b\in\mathbb{F}_{q}^{t-j},z\in\mathbb{F}_{q}$,
define

\[
f_j(a,b,z)=\sum_{\substack{pref_j(u)=a\\
\langle\mathrm{suff}_{t-j}(u),b\rangle+w=z\\
\\
}
}y_{u,w},
\]

\noindent where $\mathrm{pref}_{i}(u)$ denotes the length-$i$ prefix
vector of $u$, and $\mathrm{suff}_{i}(u)$ denotes the length-$i$
suffix vector of $u$. Then, we would like to compute

\[
\tilde{x}_{v}=\alpha\left(\sum_{u,w:\left\langle u,v\right\rangle +w=0}y_{u,w}\right)+\beta\sum_{u,w}y_{u,w}=\alpha\sum_{w} f_0(\bot,v,0)+\beta n,
\]

\noindent for all $v\in \mathbb{F}_{q}^{t}$, where $\bot$ denotes the length-$0$
empty vector. We next observe that $f$ satisfies a recurrence relation,
so that we can compute the full array of values $f_0(\bot,v,w)$
efficiently using dynamic programming and then efficiently obtain
$\ensuremath{\tilde{x}\in\mathbb{R}^{k}}$. We have

\begin{align*}
    f_j(a,b,z)&=\sum_{\substack{pref_j(u)=a\\
\langle\mathrm{suff}_{t-j}(u),b\rangle+w=z\\
\\
}
}y_{u,w}\\
&= \sum_{i=0}^{q-1}\sum_{\substack{pref_{j+1}(u)=a\circ i\\
\langle\mathrm{suff}_{t-j-1}(u),\mathrm{suff}_{t-j-1}(b)\rangle+w=z-i\cdot b_1 \pmod q\\
\\
}
}y_{u,w}\\
&= \sum_{i=0}^{q-1} f_{j+1} (a\circ i, \mathrm{suff}_{t-j-1}(b), (z-i\cdot b_1) \bmod q)
\end{align*}

Note that we have the base cases $f_t(a,\bot,w)=y_{a,w}$. We need to compute the values of $f_j(a,b,z)$ for $j\in\{0,1\ldots,t-1\},a\in\mathbb{F}_{q}^{j},b\in\mathbb{F}_{q}^{t-j},z\in\mathbb{F}_{q}$ and each value takes $O(q)$ time so the total running time is $O(tq^{t+2})$.

\section{The public coin setting}\label{sec:pub-coin}
We show that versions of \pg and \hpg can be implemented in the public coin setting in a way that the communication is $\lceil \log_2 q\rceil = \eps \log_2 e + O(1)$ bits, which is asymptotically optimal to achieve asymptotically optimal utility loss \cite[Corollary 7]{BarnesHO20}. We begin with \pg.

Recall that as described, \pg associates each of the $k$ input values with a canonical vector in $\F_q^t$. In the public coin variant we now describe, we further assume that the canonical vectors have a non-zero last coordinate. This can be ensured by picking $q,t$ such that $k \leq 1 + (1-1/q)((q^t-1)/(q-1)-1) = q^{t-1}$. We will use $C_{q,t}$ to denote the set of canonical vectors in $\F_q^t$ and $C^*_{q,t}$ to denote those with a non-zero last coordinate.

With this setup, recall that each output in the set $S_v$ can be associated with a vector $u \in C_{q,t}$ such that $\left\langle u,v\right\rangle =0$. Thus a user with input $v$ sends a vector $u \in C_{q,t}$ with probability $e^{\eps}p$ if $\left\langle u,v\right\rangle =0$ and with probability $p$ otherwise. For a vector $u$, let $pref_{t-1}(u)$ denote its length $(t-1)$ prefix. Note that for a vector $u\in C_{q,t}$, either $pref_{t-1}(u)$ is itself a canonical vector in $C_{q,t-1}$, or $u=u^* \stackrel{def}{=} (0,\ldots,0,1)$. Also note that for any $v \in C^*_{q,t}$, $u^* \neq S_v$.

This then suggests the following algorithm. We use public randomness to select a vector $w \in \F_q^{t-1}$ such that $w=(0,\ldots,0)$ with probability $p$, and $w$ is a random vector in $C_{q,t-1}$ otherwise. Thus there are $1+\frac{q^{t-1}-1}{q-1}$ possible values of $w$. Given a $w \in C_{q,t-1}$ and a $v \in C^*_{q,t}$, there is a unique $a \in \F_q$ such that $\left\langle v,w\cdot a\right\rangle = 0\mod q$. When $w \neq (0,\ldots,0)$, a user with input $v \in C^*_{q,t}$ sends message $a$ with probability $\frac{e^{\eps}}{e^{\eps}+q-1}$ if $\left\langle v,w\cdot a\right\rangle = 0\mod q$, and with probability $\frac{1}{e^{\eps}+q-1}$ otherwise. If $w=(0,\ldots,0)$, the user always send $1$.

The server given $w$ derived from the shared public randomness, and the message $a \in \F_q$, decodes it as
\begin{align*}
    Dec(w, a) &= w \cdot a.
\end{align*}
We claim that the distribution of $Dec(w, a)$ is identical to the output in the private coin \pg. First observe that by construction, $Dec(w, a) \in C_{q, t}$. Next notice that for any $u, u' \in S(v)$, we have
\begin{align*}
    \Pr(Dec(w,a)=u) &= \Pr(w = pref_{t-1}(u)) \cdot \Pr(a=u_t \mid w=pref_{t-1}(u))\\
    &= \Pr(w = pref_{t-1}(u)) \cdot\frac{e^{\eps}}{e^{\eps}+q-1}\\
    &= \Pr(w = pref_{t-1}(u')) \cdot\frac{e^{\eps}}{e^{\eps}+q-1}\;\;\;\;\; \mbox{(by uniformity of w over canonical vectors)}\\
    &= \Pr(w = pref_{t-1}(u')) \cdot \Pr(a=u'_t \mid w=pref_{t-1}(u'))\\
    &= \Pr(Dec(w,a)=u').
\end{align*}
Similarly, for any $u, u' \in C_{q,t} \setminus S_v$ such that $u, u' \neq u^*$, we can write
\begin{align*}
    \Pr(Dec(w,a)=u) &= \Pr(w = pref_{t-1}(u)) \cdot \Pr(a=u_t \mid w=pref_{t-1}(u))\\
    &= \Pr(w = pref_{t-1}(u)) \cdot\frac{1}{e^{\eps}+q-1}\\
    &= \Pr(w = pref_{t-1}(u')) \cdot \frac{1}{e^{\eps}+q-1}\;\;\;\;\; \mbox{(by uniformity of w over canonical vectors)}\\
    &= \Pr(w = pref_{t-1}(u')) \cdot \Pr(a=u'_t \mid w=pref_{t-1}(u'))\\
    &= \Pr(Dec(w,a)=u').
\end{align*}
Further, an identical calculation shows that for $u\in S_v, u' \in C_{q,t}\setminus S_v$ with $u' \neq u^*$,  $\Pr[Dec(w,a)=u] = e^{\eps} \cdot  \Pr(Dec(w,a)=u')$. Moreover, the distribution of $w$ ensures that $\Pr(Dec(w,a) = u^*) = p$. It follows that for all $u\in C_{q,t}$,  $\Pr(Dec(w,a)=u)$ is $e^{\eps}p$ if $u \in S_v$ and $p$ if $u \in C_{q,t} \setminus S_v$.

In other words, we have shown how to simulate the output distribution of \pg in the public coin setting while sending only a single element from $\F_q$.

An implementation of \hpg in the public coin model is similar. A message in \hpg is a pair $(j,u)$ where $j\in\{1,\ldots,h\}$ is the index of a block, and $u\in\mathbb F_q^t$ is the name of a canonical vector, and as above in the public coin setting we will forbid $u$ from being the all-zeroes vector (so that now we need $hq^{t-1} \ge k$). As described in \cref{sec:tradeoff}, $h,q$ are chosen so that $hq \approx e^\eps + 1$. In the public coin model, the user selects $j$ using private randomness and sends it explicitly then uses the \pg public coin protocol described above to determine the first $t-1$ entries of $u$ with no communication required, then sends the final entry of $u$ to obey the \hpg distribution. The total communication is $\lceil hq\rceil = \eps\log_2 e + O(1)$ bits.

\end{document}